\documentclass[11pt]{article}
\usepackage{amsmath,amssymb,amsthm,amsfonts,latexsym,xspace,enumerate,graphicx,color,eucal,upgreek,mathtools}
\usepackage{float}
\usepackage{balance}

\newtheorem{theorem}{Theorem}[section]
\newtheorem*{namedtheorem}{\theoremname}
\newcommand{\theoremname}{testing}

\newtheorem{lemma}[theorem]{Lemma}

\newtheorem{claim}[theorem]{Claim}

\newtheorem{fact}[theorem]{Fact}

\newtheorem*{question*}{Question}

\theoremstyle{definition}
\newtheorem{definition}[theorem]{Definition}

\theoremstyle{plain}
\newtheorem{Alg}{Algorithm}

\usepackage[backref,colorlinks,bookmarks=true]{hyperref}

\newcommand{\ignore}[1]{}

\newcommand{\defeq}{\stackrel{\small \mathrm{def}}{=}}

\newcommand{\poly}{\mathrm{poly}}

\newcommand{\R}{\mathbb R}
\newcommand{\C}{\mathbb C}

\makeatletter
\floatstyle{ruled}
\newfloat{fragment}{H}{lop}
\floatname{fragment}{Algorithm}
\renewcommand{\floatc@ruled}[2]{\vspace{2pt}{\@fs@cfont \#1.\:} \#2 \par
 \vspace{1pt}}
\makeatother

\title{A Polynomial Time Algorithm for Lossy Population Recovery}
\author{Ankur Moitra \thanks{
Institute for Advanced Study and Princeton University. Email: {\tt moitra@ias.edu}.
Research supported in part 
by NSF grant No.
DMS-0835373 and by an NSF Computing and Innovation Fellowship.} \and Michael Saks \thanks{Rutgers University. Email: {\tt saks@math.rutgers.edu}. Research supported in part by 
NSF under grants CCF-083727 and CCF-1218711, and by a sabbatical at Princeton University.}}

\begin{document}
\maketitle

\begin{abstract}
We give a polynomial time algorithm for the lossy population recovery problem.  In this problem, the goal
is to approximately learn 
an unknown distribution on binary strings of length $n$ from lossy samples: for some parameter $\mu$ each coordinate of the sample is preserved with probability $\mu$ and
otherwise is replaced by a `?'.  The running time and number of samples needed for our algorithm is polynomial in $n$ and $1/\varepsilon$
for each fixed $\mu>0$. This improves on
the algorithm
of Wigderson and Yehudayoff \cite{WY} that runs in quasi-polynomial time for any fixed $\mu>0$ and the polynomial time algorithm of Dvir et al \cite{DRWY} 
which was shown to work for $\mu \gtrapprox 0.30$ by  Batman et al \cite{BIMPS}. 
In fact, our algorithm also works in the more general framework of Batman et al. \cite{BIMPS} in which there is no
a priori bound on the size of the support of the distribution.  
The algorithm we analyze is implicit in previous work \cite{DRWY,BIMPS}; our main contribution is to
analyze the algorithm by showing (via linear programming duality and connections to complex analysis) that a certain matrix associated with the problem has
a robust local inverse even though its condition number is exponentially small.  A corollary of our result is the first polynomial time algorithm for learning DNFs in the restriction access model of Dvir et al \cite{DRWY}.

\end{abstract}

\thispagestyle{empty}

\newpage

\setcounter{page}{1}

\section{Introduction}

\subsection{Background and Our Results}

The population recovery problem was introduced by Dvir et al \cite{DRWY} and also studied by Wigderson and Yehudayoff \cite{WY}. To describe this basic statistical problem, we will borrow an example from \cite{WY}: 

\begin{quote}
Imagine that you are a paleontologist, who wishes to determine the population of dinosaurs that roamed the Earth before the hypothesized meteor made them extinct. Typical observations of dinosaurs consist of finding a few teeth of one here, a tailbone of another there, perhaps with some more luck a skull and several vertebrae of a third, and rarely a near complete skeleton of a fourth....Using these fragments, you are supposed to figure out the population of dinosaurs, namely a complete description of (say) the bone skeleton of each species and the fraction of each species occupied in the entire dinosaur population.
\end{quote}

To make this precise, suppose there is an unknown distribution $\pi$ over binary strings of length $n$.
We are given samples from the following model: 

\begin{itemize}

\item Choose a string $a$ according to $\pi$

\item Replace each coordinate with a ``?" independently with probability $1 - \mu$

\end{itemize}

\noindent The goal is to reconstruct the distribution up to an additive error $\varepsilon>0$.  
We would like to output a set of strings $S$ and for each string $a$ in $S$, an estimate $\widetilde{\pi}(a)$
of $\pi(a)$ with the requirement that each of these estimates is within $\varepsilon$ of $\pi(a)$,
and for every string  $a \not\in S$, $\pi(a)$ must be at most $\varepsilon$.  This formulation of the problem is adapted from  \cite{BIMPS}; in
the original version in \cite{DRWY} the support size of the distribution (which we will denote by $k$) is also a parameter.

We remark that the maximum likelihood estimator can be computed efficiently using a convex program \cite{BRW}. Yet the challenge is in showing that few samples are needed {\em information theoretically}.  We will see another instance of this type of issue in our paper: our approach is based on the notion of a `robust local inverse' (defined later) \cite{DRWY}, which is easy to compute and the challenge is in showing that a good robust local inverse {\em exists} for any fixed $\mu > 0$.

Dvir et al \cite{DRWY} gave a polynomial time algorithm for lossy population recovery for any $\mu \gtrapprox 0.365$; their analysis was improved
by Batman, et al. \cite{BIMPS} who showed that the same algorithm works for any $\mu> 1- 1/\sqrt{2} \approx 0.30$.
Wigderson and Yehudayoff \cite{WY} gave an alternate approach based on a method termed ``partial identification" that runs in time quasi-polynomially in the support size $k$ for any fixed $\mu > 0$.
Interestingly, Wigderson and Yehudayoff \cite{WY} show that their framework cannot be used to get a polynomial time algorithm (and the number of samples needed is at least $k^{\log\log k}$). In fact, their algorithm works even in the presence of corruptions, not just erasures (whereas ours does not).  

A generalization of the population recovery problem was introduced in the seminal work of Kearns et al \cite{KMRRSS}, which they called the problem of learning mixtures of Hamming balls: Again, we choose a string $a$ according to $\pi$ but now each bit in $a$ is flipped with probability $\eta_a < 1/2$ and this probability is allowed to depend on $a$. Kearns et al \cite{KMRRSS} give algorithms for the special case in which each flip probability is the same (which is exactly the noisy population recovery problem) and their algorithms run in time exponential in the support size $k$. This is an interesting phenomenon in learning distributions, that for many problems we do not know how to achieve a running time that is sub-exponential in the number of components. For example, this is the case when learning mixtures of product distributions \cite{FOS}, learning juntas \cite{MOS}, learning decision trees \cite{EH}, and learning mixtures of Gaussians \cite{MV}, \cite{BS}. In fact, many of these problems are inter-reducible \cite{FGKP}, so with this context it is interesting that the population recovery problem is one positive example where we can avoid exponential dependence on $k$. But is there a polynomial time algorithm? 

Here we give an estimator that solves the population recovery problem, such that for any fixed $\mu>0$,
the running time and number of samples needed is polynomial in $n$ and  $1/\varepsilon$.

\begin{theorem}
\label{pop rec alg}
There is an efficient algorithm for the population recovery problem whose running time and number of samples needed is
 $O((n/\varepsilon )^{2f(\mu)} )$ where $f(\mu) = 1/\mu \log 2/\mu + O(1)$.
\end{theorem}

The population recovery problem arose naturally from the investigation of one of the central problems in
learning theory, 
learning  DNFs.  The best known algorithm for learning DNFs in Valiant's PAC learning model \cite{V} runs in time (roughly) $2^{n^{1/3}}$ \cite{KS}. 
Dvir et al \cite{DRWY} introduced a new model called restriction access that can be thought of as an interpolation between black box and white box access to the function: Each example consists of a restriction of the unknown DNF obtained by fixing a random $1-\mu$ 
fraction of the input variables.. Dvir et al \cite{DRWY} showed how to reduce the problem of learning an
$n$ term, $k$-variable DNF  to solving an instance of the population recovery problem on strings of length $n$ and support size $k$.  Previous
 algorithms for population recovery yield a polynomial time algorithm for learning DNFs in the restriction access model for any $\mu > 1-1/\sqrt{2}$ \cite{DRWY,BIMPS}, and Wigderson and Yehudayoff \cite{WY} obtain a quasi-polynomial time algorithm that runs in time $k^{\log k}$ (where $k$ is the number of clauses). Combining the reduction of \cite{DRWY} with
Theorem \ref{pop rec alg} immediately gives:

\begin{theorem}
\label{dng alg}
There is an efficient algorithm to PAC learn DNFs in the restriction access model for any $\mu > 0$. The running time and number of samples needed is $O((n/\varepsilon)^{2f(\mu)} \poly(n,k))$ where $f(\mu) = 1/\mu \log 2/\mu$ and the algorithm succeeds with high probability. 
\end{theorem}

\noindent The main open question in this paper is whether there is a polynomial time algorithm for {\em noisy} population recovery (see Section~\ref{sec:open}). If the goal was to learn a distribution that is close to the true distribution (rather than more stringent goal of learning its parameters), the maximum likelihood estimator would suffice. But the main obstacle here is showing that for two distributions whose parameters do not match, their statistical distance is noticeably large.

\subsection{The Robust Local Inverse}
\label{RLI}
At a more philosophical level, what makes the population recovery problem particularly interesting is that in order to give an efficient estimator we need to solve a certain inverse problem despite the fact that the corresponding matrix has many exponentially small eigenvalues. 
As will be reviewed in the next section, Dvir et al \cite{DRWY} showed that  the population
recovery problem can be reduced to a problem of the following form:  We have two unknown probability distributions
$\pi$ and $\phi$ over the domain $\{0,\ldots,n\}$, which when viewed as vectors indexed by
$\{0,\ldots,n\}$ are related by the equation:
$$\phi^T=\pi^TA$$
where $A$ is a known (row) stochastic invertible matrix indexed by $\{0,\ldots,n\}$.
We want to estimate $\pi(0)$ but we only have access to samples chosen according to the
distribution $\phi$.  We would like the running time and number of samples needed to be (at most) a polynomial in $n$ and $1/\varepsilon$. 

Let $u$ denote the first column of $A^{-1}$. Then: $$\pi(0)=\phi^T u$$
So, if we knew the vector $\phi$ {\em exactly}, we could use it to recover $\pi(0)$ exactly.
But we do not know $\phi$. We can  estimate $\phi$ from random random samples in the obvious way: let $\widetilde{\phi}(j)$ be the fraction of observed samples having exactly $j$ zero entries.
We might then hope that $\widetilde{\pi}(0)=\widetilde{\phi}^T u$ is a good estimate to $\pi(0)$.  We will refer to this
as the {\em natural estimator} for $\pi(0)$.
The 
error $|(\widetilde{\phi}-\phi)^T u|$ of this estimator is at most  $n\|\widetilde{\phi}-\phi\|_{\infty} \|u\|_{\infty}$. 
Thus to obtain estimation error  $\varepsilon$, it is enough that

\[
\|\widetilde{\phi}-\phi\|_{\infty} \leq \frac{\varepsilon}{n\|u\|_{\infty}}
\]
and the Chernoff-Hoeffding bound says that  $C\log(n)(\|u\|_{\infty}n /\varepsilon)^2$  samples are enough so that
the probability of exceeding the desired error is less than $e^{-C}$.

To ensure that this is not too many samples, we want that $\|u\|_{\infty}$ is polynomially bounded.
For the estimation problem that is derived from population recovery, it turns out that $\|u\|_{\infty} \leq 1$ provided that $\mu \geq 1/2$
and so in this case the natural estimator yields an accurate estimate from a polynomial number of samples. But if $\mu < 1/2$, then $\|u\|_{\infty}$ is
exponentially large in $n$ and this estimator requires exponentially many samples to be at all accurate.

What else can we do? The vector $u$ has entries that are too large, so Dvir et al \cite{DRWY} suggested replacing $u$ by another vector $v$ whose entries are not too
large and such that $\pi^TA v $ is close to $\pi^TA u$
for all distributions $\pi$.  Remarkably, Dvir et al \cite{DRWY} managed to construct such a $v$ which works for $\mu  \gtrapprox .365$ (the analysis was subsequently improved to $\mu \gtrapprox .3$ \cite{BIMPS}), which in turn yields a polynomial time algorithm
for the population recovery problem even in cases when the natural estimator fails! 

Since $\pi^T A u = \pi(0)$, it follows that what we really want is to find a vector $v$ so that
$$\|A v - e_0\|_{\infty} \leq \varepsilon$$
where $e_0$ is the indicator vector for zero (i.e. its first entry is one and the rest are all zero). And furthermore we want $\|v\|_\infty$ to be as small as possible. 
A vector satisfying the above condition is called a  {\em $\varepsilon$-local inverse} for $A$ at $e_0$, and we will refer to $\|v\|_{\infty}$ as the
{\em sensitivity} of $v$.    
If we can find a $v$ whose sensitivity is at most $\sigma$, then $poly(n,1/\varepsilon,\sigma)$ samples
suffice to get an estimate ($\widetilde{\phi}^T v$) to $\pi(0)$ that is within an additive $\varepsilon$.

Geometrically, a local inverse is obtained by taking $A^{-1}e$ where $e$ is a small  perturbation of the vector $e_0$, which
is chosen so that  $A^{-1}e$ has small norm even though $A^{-1} e_0$ does not. What controls the  behavior of
$A^{-1}e$ is the representation of  $e$ in the basis of singular vectors of $A$.  In choosing $e$ we want to
remove from $e_0$ the components corresponding to tiny singular values, which will ensure that the sensitivity of $v=A^{-1}e$ is not too large.
We are hoping that that the weight on these deleted components is small so that
the result is a good local inverse.

The problem of finding the $\varepsilon$-local inverse of minimum sensitivity for a particular matrix $A$ can be expressed directly as a linear program whose variables are the vector $v$ and the sensitivity $\sigma$:
\begin{eqnarray}
\label{LP}
& \min \sigma & \\
\nonumber
Av & \geq & e_0-\varepsilon \bold{1} \\
\nonumber
 -Av &  \geq & -e_0-\varepsilon \bold{1}\\
\nonumber 
v+\sigma \bold{1} & \geq & \bold{0}\\
\nonumber 
-v + \sigma \bold{1} & \geq & \bold{0}
\end{eqnarray}
The solution $v$ can be used in to estimate $\pi(0)$ from $\widetilde{\phi}$, where the number of samples depends on $\sigma$, as above. Note that the matrix $A$ depends on $\mu$. Our main contribution is to prove that there is a good solution to the above linear program for any $\mu > 0$.

The approach in Dvir et al \cite{DRWY} and in Batman et al \cite{BIMPS} was to guess a solution to the above linear program and bound its sensitivity. Instead, we consider the dual (maximization) problem
and prove an upper bound on its maximum.    
After some work, the dual problem becomes a problem of finding a polynomial $p$ of degree $n$ so as to  maximize $p(0)-\varepsilon\|p\|_1$, where $\|\cdot\|_1$ denotes the sum of the absolute values of the coefficients, 
subject to the constraint that the translated polynomial $q(x)=p(1+(x-1)/\mu)$ has $\|q\|_1=1$.  Bounding this maximum from above is then reduced to a problem of
showing that if $p$ is a polynomial (indeed any holomorphic function)
on the complex plane and there exists a disk of nontrivial diameter where $|p(z)|$ is much smaller than $|p(0)|$ then
the maximum of $p(z)$ on the unit circle must be much larger than $|p(0)|$.  This final result can be viewed as a kind of uncertainty principle 
and is  proved using tools from complex analysis (the Hadamard 3-circle theorem, and the M\"obius transform).

\section{Reductions for Population Recovery}\label{sec:rli}

Here  we describe (informally) the reduction of Dvir, et al. \cite{DRWY} from the population recovery problem to the problem of constructing
a robust local inverse for a certain matrix $A$ (whose entries depend on $\mu$): Recall that if we choose a string $a \in \{0,1\}^n$ (according to $\pi$), the
observation is a (random) string in $\{0,1,?\}^n$ obtained from $a$ by
 replacing each $a_i$ with `?' independently with probability $1-\mu$.  
 
 The first observation of Dvir et al \cite{DRWY} is that we may as well assume that we know all of the strings $a$ whose probability $\pi(a)$ is at least $\Omega(\varepsilon)$. 
Of course, in the population recovery both the strings and their probabilities are unknown, so how can we
 reduce from the case when everything is unknown to the case where at least the set of strings with large probability is known? 
Suppose we ignore all but the first $n'$ coordinates; then we get an instance of the population recovery problem on length $n'$ strings. In particular, the
probability $\pi(a')$ of a length $n'$ string is the total probability of all length $n$ strings $a$ whose first $n'$ coordinates are exactly $a'$. 
Now the rough idea is that we can incrementally solve the population recovery problem on longer and longer prefixes, each time we increase the length of the
prefix by one we at most double the number of candidate strings. The crucial insight is that we can always prune the set of strings because we never need to
keep a prefix whose total probability is less than $\varepsilon$. 

The second observation of Dvir et al \cite{DRWY} is that if all the strings are known, then it suffices to estimate $\pi(0)$ within an additive $\varepsilon$. This type of reduction
is standard: given a string $a$, we can take each observation and XOR it with $a$ but keeping the symbol `?' unchanged. The samples we are given can be thought of
as samples from an instance of population recovery where every string is mapped to its XOR with $a$, and so we can recover $\pi(a)$ by finding the probability of the all zero
string in this new instance of the problem. 

The final simplification is: suppose we ignore the locations of the ones, zeros and question marks in the samples but only recover the number of ones. Then we can map the probability
distribution $\pi$ to a length $n + 1$ vector where $\pi(i)$ is the total probability of all strings with exactly $i$ ones. What is the probability that we observe $j$ ones (and the remaining symbols are
zeros or question marks) given that the sample $a$ had $j$ ones? This quantity is exactly:
 $$
A_{i, j} = {i \choose j} \mu^{j} (1-\mu)^{i-j}
 $$
 So if we only count the number of ones in each observation, we are given random samples from the distribution $\pi^T A$. Hence, if our goal is to recover the probability $\pi(0)$ assigned to the all zero string, and we ignore {\em where} the zeros, ones and question marks occur in our samples, we are faced with a particular matrix $A$ (whose entries depend on $\mu$) for which we would like to construct a robust local inverse. 
 
 \begin{definition}
 Let $\sigma_n(\mu,\varepsilon)$ denote the minimum sensitivity of a $\varepsilon$-local
inverse (i.e. the optimum value of (\ref{LP})). 
 \end{definition}
 
The following family of vectors will play a crucial role in our analysis: 
$$v^{\alpha}= \Big [1, \alpha, \alpha^2, ... \alpha^{n-1}\Big ]$$
Then it can be checked that setting $\alpha = -\frac{1-\mu}{\mu}$ is the natural estimator (i.e. $v^{\alpha} = A^{-1}e_0$) and the sensitivity of this estimator is exponentially large
for $\mu < 1/2$. We prove:

\begin{theorem}
\label{sigma bound}
For all positive integers $n$ and $\mu,\varepsilon >0$ we have $\sigma_n(\mu,\varepsilon) \leq (1/\varepsilon)^{f(\mu)}$
where $f(u)=\frac{1}{\mu} \log \frac{2}{\mu}$.
\end{theorem}

Theorem \ref{pop rec alg} follows since as discussed in Section \ref{RLI}, the number of samples we 
need to obtain the desired approximation with high probability when using the best local
inverse is $\sigma_n(\mu,\varepsilon)^2 \poly(n,1/\varepsilon)$.

\section{A Transformed Linear Program}

As outlined earlier,  the problem of finding an $\varepsilon$-local inverse can be expressed
as a linear programming problem whose objective is to minimize the sensitivity.  We want to prove an upper bound
on the value of the solution, and we will accomplish this by instead bounding the maximum objective function of the dual. 

However, before passing to the dual we will apply a crucial change of basis to the linear program.  The reason we do this is so that the dual can then be interpreted
as a certain maximization problem over degree $n$ polynomials. 
We will choose $n+1$ values $\alpha_0, \alpha_1, \ldots, \alpha_n$ (as we'll see the particular values  won't matter) 
and we will consider the estimators $v^{\alpha_i}$ defined in the previous section. We will abuse notation and refer to this estimator
as $v^i$. Since this family forms a basis, we can write any local inverse $v$ in the form $v = \sum_{i=0}^n \lambda_i v^i$.
Let $V$ be the columns $v^0,\ldots,v^n$  and let $B=AV$. Then our new linear program is:
\begin{eqnarray}
\label{LP}
& \min \sigma & \\
\nonumber
B \lambda &  \geq & e_0-\varepsilon \bold{1}\\
\nonumber
-B \lambda & \geq & - e_0 +\varepsilon \bold{1} \\
\nonumber 
V \lambda + \sigma \bold{1} & \geq & \bold{0} \\
\nonumber 
-V\lambda + \sigma \bold{1} & \geq & \bold{0}\\
\nonumber 
\sigma & \geq & 0
\end{eqnarray}
The final constraint is superfluous, but is helpful in formulating the dual linear program.

The coefficient matrix $V$ is a Vandermonde matrix (i.e. each column has the form $v^{\alpha}$ for some $\alpha$)
with the entry in row $i$ and column $j$ given by
$V_i^j=(\alpha_j)^i$ (with $V_0^0=1$).  In fact, it turns out that $B$ is also a Vandermonde matrix whose $j^{th}$ column
is exactly $v_{1+\mu(\alpha_j-1)}$:
 $$B_i^j = \sum_{k \leq i} {i \choose k} \mu^k (1-\mu)^{i-k} (\alpha_j)^k= (1-\mu)^i \sum_{k \leq i} {i \choose k} \Big (\frac{\alpha_j \mu}{1 - \mu} \Big )^k = (1-\mu)^i (1 + \frac{\alpha_j \mu}{1 - \mu})^i = (1 + \mu(\alpha_j-1))^i$$ 
Indeed, this simple form for $B$ is precisely the reason we chose this basis transformation.

The new linear program has $n+2$ variables and $4(n+1)$ constraints (consisting of four groups of $n+1$ constraints each)
so the dual will have $4(n+1)$ variables consisting of four vectors, denoted by $p^+,p^-,q^-,q^+$ each
indexed by  $\{0,\ldots,n\}$.The resulting dual program is:
\begin{eqnarray}
\label{DLP}
 \max p^+_0-p^-_0-\varepsilon \sum_i p^+_i+p^-_i && \\
\nonumber
(p^+ - p^-)^T B  + (q^- -q^+) V &  = & \bold{0}\\
\nonumber
\sum_i (q^+_i +q^-_i) & \leq  & 1\\
\nonumber 
p^+,p^-,q^+,q^- &\geq& \bold{0} 
\end{eqnarray}

We can now make some simplifying observations.
If for any $i$, both $p^+_i$ and $p^-_i$ are positive, we can decrease them
each by their minimum without violating the constraints, and only increasing the objective function.  
So we may assume that at least one of them is zero.
Similarly for $q^+_i$ and $q^-_i$.  Then we can define $p=p^+ - p^-$ and $q=q^+-q^-$ to simplify the dual linear program
to:
\begin{eqnarray}
\label{DLP1}
 \max p_0-\varepsilon \sum_i |p_i| && \\
\nonumber
p^T B & = & q^T V \\
\nonumber
\sum_i |q_i| & \leq  & 1.
\end{eqnarray}

Define the polynomials $p(x)=\sum_{j=0}^n p_j x^j$ and $q(x)=\sum_{j=0}^n q_j x^j$.
The equality constraint gives $n+1$ equations indexed from 0 to $n+1$
where $j^{th}$ constraint is that $p(1+\mu(\alpha_j-1))=q(j)$.  Since $p(1+\mu(x-1))$ and $q(x)$ agree
on $n+1$ values they must be the same polynomial.    This leads to the following formulation:
\begin{quote}
The optimal sensitivity $\sigma_n(\mu,\varepsilon)$ is equal to the maximum of $p(0)-\varepsilon \|p\|_1$ over all
degree $n$ polynomials for which the translated polynomial
$q(x)=p(1+\mu(x-1))$
satisfies $\|q\|_1 \leq 1$.
\end{quote}
Recall that $\| p \|_1$ denotes the sum of the absolute values of the coefficients.

So now our goal is to prove an upper bound on the maximum of this linear program.  
 We can think of this as trying to show a type of {\em uncertainty principle} for the coefficients of a polynomial when applying an affine change of variables. There is a considerable amount of literature on establishing uncertainty principles for functions and their Fourier transforms (see e.g. \cite{DS}), but there seems to be no literature concerning other affine changes of variables (i.e. $p(1+\mu(x-1)) = q(x)$). In fact, here we will establish such an uncertainty principle via the Hadamard three circle theorem in complex analysis.

\section{Sup Relaxations}

The quantities $\|p\|_1$ and $\|q\|_1$ are unwieldy - e.g. given just the graph of the polynomial, what can we say about its coefficients? Here we will relax constraints on $\|p\|_1$ by instead considering the maximum of the polynomial over certain domains.

\begin{definition}
({\em Restricted $\sup$-norm}.)
For a subset $W$ of $\R$, let $\|q\|_{\sup}^{W} \defeq \sup_{x \in W} |q(x)|$. 
\end{definition}

Recall that we used the notation $\|q\|_1$ to denote the sum of the absolute values of the coefficients of $q$. Then it is easy to see that:

\begin{claim}
$\|q\|_1 \geq \|q\|_{\sup}^{[-1,1]}$
\end{claim}

\begin{proof}
For each $x \in [-1,1]$,
$\sum_{i = 0}^n |q_i| |x|^i \leq \sum_{i = 0}^n |q_i| = \|q(x)\|_1$.
\end{proof}

In the polynomial formulation of $\sigma_n(\mu,\varepsilon)$ replacing the objective
function by $p(0)-\varepsilon \|p\|_{\sup}^{[-1,1]}$ only increases the value of the objective function.
Similarly, replacing the constraint $\|q\|_1 \leq 1$
by $\|q\|_{\sup}^{[-1,1]} \leq 1$ can only increase the objective function.  
Since $q(x)=p(1+\mu(x-1))$ and the transformation $x \longrightarrow 1+\mu(x-1)$ maps
the interval $[-1,1]$ to the interval $[1-2\mu,1]$ we have $\|q\|_{\sup}^{[-1,1]}=\|p\|^{[1-2\mu,1]}$.
This leads to a relaxation of the polynomial formulation:

\begin{quote}
The optimal sensitivity $\sigma_n(\mu,\varepsilon)$ is at most  the maximum of  $p(0)-\varepsilon \|p\|_{\sup}^{[-1,1]}$
over all degree $n$ polynomials $p(x)$ for which $\|p\|_{\sup}^{[1-2\mu,1]} \leq 1$.
\end{quote}

For this relaxation to be useful to us we will need to prove that the new objective
function can not be too large if $p$ satisfies
the constraints of the relaxation.
Informally, we will say that a polynomial is bad if it satisfies the constraints of the relaxation
and makes the objective function very large. If $\mu \geq 1/2$ then $0 \in [1-2\mu,1]$ 
and so $|p(0)| \leq 1$ and no polynomial can be bad.  So assume $\mu < 1/2$.
A bad polynomial 
must be bounded  between $-1$ and $1$ on the interval $[1-2\mu,1]$,  must be very large at the origin, 
and must have $|p(x)|$  at most  $|p(0)|/\varepsilon$ for all $x \in [-1,1]$.

Is there any polynomial that satisfies these conditions? Unfortunately for this approach, there is.  The polynomial
$(1-x^2)^{n/2}$ has its maximum on $[-1,1]$ at the origin, where it is $1$, and its maximum on $[1-2\mu,1]$ is at $1-2\mu$
where its value is $C=(2\mu-\mu^2)^{n/2}$ which is exponentially small.  Thus the polynomial $p(x)=\frac{1}{C}(1-x^2)^{n/2}$
satisfies the constraints and has objective function value that is exponentially large
in $n$.

To salvage this approach we move to complex numbers. The definition of the restricted $\sup$-norm  
extends directly to subsets $W$ of the complex numbers.
For $\beta \in \C$ and positive real number $\gamma$ let $D_{\gamma}(\beta)$ be the closed disk in the complex plane
of radius $\gamma$ centered at $\beta$.   Let $C_{\gamma}(\beta)$ be the circle bounding
$D_{\gamma}(\beta)$.   If $\beta=0$ we write simply $D_{\gamma}$ and $C_{\gamma}$.
As with

\begin{claim}
$\|q(x)\|_1 \geq \|q(x)\|_{\sup}^{D_1}$.
\end{claim}

Observe that the image of the disk $D_1$ under the transformation $x \longrightarrow (1+\mu(x-1))$
is $D_{\mu}(1-\mu)$.    Just as before we obtain the following relaxation:

\begin{quote}
The optimal sensitivity $\sigma_n(\mu,\varepsilon)$ is at most  the maximum of  $p(0)-\varepsilon \|p\|_{\sup}^{D_1}$
over all degree $n$ polynomials $p(x)$ such that $\|p\|_{\sup}^{D_{\mu}(1-\mu)} \leq 1$
\end{quote}

As we will see,  there are no bad polynomials for this relaxation.
In hindsight, it is not surprising that the values of the polynomial $p(x)$ over the whole complex disk reveals 
much more information than just the values on $[-1,1]$; in particular, we can recover the values of a polynomial from integrating around the circle, so a polynomial cannot stay too small on the boundary of the disk if it is large at the origin.
In particular the polynomial $(x^2 -1)^{n/2}$ that was bad for the
$\| \cdot \|^{[-1, 1]}_{\sup}$ relaxation is no longer bad because its maximum (on $D_1$) is attained at $x=\imath$ and is exponentially large. 

In the next section we will prove:

\begin{lemma}
\label{holo}
Let $h$ be a holomorphic function and suppose $D=D_{\rho}(\beta)$ is a disk contained in $D_1$
such that $\|h\|_{\sup}^D \leq 1$.  Then there is a point $x \in C_1$
such that $|h(x)| \geq |h(0)|^{1+d}$, where $d=(1-|\beta|)/\log(2/\rho)$.
\end{lemma}

From this uncertainty principle, we can now prove Theorem \ref{sigma bound}.

\vspace{0.5pc}

\noindent
{\bf Proof of Theorem \ref{sigma bound}. }
We use the bound from the $D_1$-sup relaxation.
Let $p$ be a polynomial
satisfying the constraints and let $s=|p(0)|$.  Then $p$ satisfies the conditions of Lemma \ref{holo} with $\beta=1-\mu$ and  $\rho=\mu$.
Therefore $\|p\|_{\sup}^{D_1} \geq |s|^{d+1}$, where $d$ is as in the lemma.  
From this we conclude that the
objective function in the $\| \cdot \|^{D_1}_{\sup}$ relaxation is at most $s-\varepsilon s^d$ which is maximized when
$s=(1/(d+1)\varepsilon)^{1/d}$, and this quantity is itself an upper bound on the objective function.
We can therefore conclude that $\sigma_n(\varepsilon,\mu) \leq (1/\varepsilon)^{1/d}$ where the exponent
is equal to $\frac{1}{\mu}\log\frac{2}{\mu}$.

\section{Proof of Lemma \ref{holo}}

Here we will prove the uncertainty principle stated in the previous section using tools from complex analysis. Perhaps one of the most useful theorems in understanding the rate of growth of holomorphic functions in the complex plane is Hadamard's Three Circle Theorem (and the related Three Lines Theorem):

\begin{theorem} \cite{H}
\label{Hadamard}
Let $0 < a \leq b \leq c$  and let  $g(x)$ be holomorphic function on the $D_{c}$.   
Then 
$$\log \frac{c}{a}\log \|g\|_{\sup}^{C_b} \leq \log \frac{c}{b} \log \|g\|_{\sup}^{C_a}+ \log \frac{b}{a} \log \|g\|_{\sup}^{C_c}.$$
\end{theorem}

In Lemma \ref{holo} we do {\em not} have three concentric circles but we can apply a M\"obius transformation to
put the problem in the right form.   Let $\beta$ be the center of the disk $D$ in the lemma
and consider the transformation $\phi(x)=\phi_{\beta}(x)=\frac{\beta+x}{1+ \beta^* x}$, where $(\cdot)^*$ denotes
complex conjugate.   The following fact is well known and easy to check:

\begin{fact}
\label{mobius}
For $|\beta| < 1$, $\phi_{\beta}$ is a holomorphic function which maps $D_1$ to itself.
\end{fact}

\begin{enumerate}
\item $\phi(C_1)=C_1$.
\item $0 \in \phi(C_{|\beta|})$.
\item $\phi(C_{\rho/2}) \subseteq D= D_{\rho}(\beta)$.
\end{enumerate}

The first claim is a standard fact about M\"obius transformations.  The second follows from $\phi(-\beta)=0$.
For the third,

$$|\phi(x)-\beta| = \Big |\frac{\beta+x}{1+\beta^* x}-\beta \Big | = \Big |\frac{x(1-|\beta|^2)}{1+\beta^* x} \Big | \leq |x|\frac{1-|\beta|^2}{1-|\beta|}=|x|(1+|\beta|)\leq 2|x|.
$$

Now consider the function $g$ defined on $D_1$ by $g(x)=h(\phi(x))$.  From the three previous observations we have:

\begin{enumerate}
\item $g(C_1)=h(C_1)$ and so $\|g\|_{\sup}^{C_1}=\|h\|_{\sup}^{C_1}$.
\item $h(0) \in g(C_{\beta})$ so $\|g\|_{\sup}^{C_{\beta}} \geq |h(0)|$.
\item $g(C_{\rho/2}) \subseteq h(D)$ so $\|g\|_{\sup}^{C_{\rho/2}}| \leq \|h\|_{\sup}^{D} \leq 1$,  by
the hypothesis of the lemma. 
\end{enumerate} 
Applying Theorem \ref{Hadamard} with 
$a=\rho/2$, $b=|\beta|$ and $c=1$ we get:

$$\log \frac{2}{\rho} \|g\|_{\sup}^{C_{|\beta|}} \leq \log \frac{1}{\beta} \log \|g\|_{\sup}^{C_{\rho/2}}+ 
\log \frac{2\beta}{\rho} \log \|g\|_{\sup}^{C_1},$$
which when combined with the three previous bounds gives:
$$\log \frac{2}{\rho} |h(0)| \leq 
\log \frac{2\beta}{\rho} \log \|h\|_{\sup}^{C_1},$$
from which we conclude:
$$
\|h\|_{\sup}^{C_1} \geq |h(0| ^t,
$$
where $t= \log(\frac{2}{\rho})/\log \frac{2\beta}{\rho} = 1+ \log(1/\beta)/\log(2\beta/\rho) \geq 1+ (1-\beta)/\log(2/\rho)$,
which is the parameter $d$ defined in the lemma.

\section{Open Question}\label{sec:open}

Is there a polynomial time algorithm for noisy population recovery -- i.e. when attributes are not deleted, but are flipped (with probability $\eta < 1/2$)? It seems that new ideas are needed to handle this case in part because if we try the same method of writing a linear program over a basis of estimators, then instead of two polynomials related by an affine change of variables, we get two polynomials $p(x)$ and $q(x)$ for which $p(x) = \ell(x)^n q(\phi(x))$ where $\ell(x)$ is a linear function and $\phi(x)$ is a M\"obius transformation. However this damping term $\ell(x)^n$ makes it much easier for $q(x)$ to be bounded in the complex disk.

\end{document}